\documentclass[a4paper,english,numberwithinsect]{basic}

\usepackage{microtype} 

\usepackage{hyperref}

\usepackage{xspace}
\usepackage{todonotes}
\usepackage{caption}
\usepackage{subcaption}
\usepackage{thm-restate}

\newcommand{\SCR}{\textsf{SCR}\xspace}
\newcommand{\aul}{$\alpha_{t\ell}$\xspace}
\newcommand{\aur}{$\alpha_{tr}$\xspace}
\newcommand{\adr}{$\alpha_{br}$\xspace}
\newcommand{\adl}{$\alpha_{b\ell}$\xspace}
\newcommand{\gur}{$g_{tr}$\xspace}
\newcommand{\gdr}{$g_{br}$\xspace}
\newcommand{\gul}{$g_{t\ell}$\xspace}
\newcommand{\gdl}{$g_{b\ell}$\xspace}
\newcommand{\su}{$s_t$\xspace}
\newcommand{\sr}{$s_r$\xspace}
\newcommand{\sd}{$s_b$\xspace}
\newcommand{\slf}{$s_\ell$\xspace}

\newcommand{\len}{\text{len}}

\def\imagebox#1#2{\vtop to #1{\null\hbox{#2}\vfill}}

  {\begin{list}{}%
          {\setlength{\leftmargin}{#1}}%
          \item[]%
  }
  {\end{list}}

\title{Axis-Aligned Square Contact Representations}
\titlerunning{Axis-Aligned Square Contact Representations}

\author[1]{Andrew Nathenson\thanks{Research on this paper was partially supported by the NSF award DMS-1800734.}}
\affil[1]{California State University Northridge\\
  \texttt{andrew.nathenson.540@my.csun.edu}}
\authorrunning{Andrew Nathenson}

\begin{document}

\maketitle

\begin{abstract}
We introduce a new class $\mathcal{G}$ of  bipartite plane graphs and prove that each graph in $\mathcal{G}$ admits a proper square contact representation.  
A contact between two squares is \emph{proper} if they intersect in a line segment of positive length. The class $\mathcal{G}$ is the family of quadrangulations obtained from the 4-cycle $C_4$ by successively inserting a single vertex or a 4-cycle of vertices into a face.  

For every graph $G\in \mathcal{G}$, we construct a proper square contact representation. The key parameter of the recursive construction is the aspect ratio of the rectangle bounded by the four outer squares. We show that this aspect ratio may continuously vary in an interval $I_G$. The interval $I_G$ cannot be replaced by a fixed aspect ratio, however, as we show, the feasible interval $I_G$ may be an arbitrarily small neighborhood of any positive real.

\end{abstract}

\section{Introduction}

Geometric representations of graphs have many applications and yield intriguing problems~\cite{Lovasz}. Koebe's celebrated \emph{circle packing theorem}~\cite{Koebe}, for example, states that every planar graph is a contact graph of interior-disjoint disks in the plane. Schramm~\cite{Schramm} proved that this theorem holds even if we replace the disks with homothets of an arbitrary smooth strictly convex body in the plane. The result extends to non-smooth convex bodies in a weaker form (where a homothet may degenerate to a point, and three or more homothets may have a common point of intersection), and every planar graph is only a \emph{subgraph} of such a contact graph.

In this paper, we consider \emph{strong} contact representations with interior-disjoint convex bodies where no three convex bodies have a point in common. 
It is an open problem to classify graphs that admit a strong contact representation with homothets of a triangle or a square~\cite{DBLP:conf/cccg/BadentBGDFGKPPT07,DBLP:conf/isaac/LozzoDEJ17}. It is known that every partial 3-tree~\cite{DBLP:conf/cccg/BadentBGDFGKPPT07} and every 4-connected planar graph admits a strong contact representation with homothetic triangles, see~\cite{DBLP:conf/compgeom/FelsnerF11,DBLP:journals/dcg/GoncalvesLP12}; but there are 3-connected planar graphs which do not admit such a representation. We note here that every planar graph admits a strong contact representation with (non-homothetic) triangles~\cite{DBLP:journals/cpc/FraysseixMR94}; see also~\cite{DBLP:journals/dcg/GoncalvesLP12}.

Strong contact representations with homothetic squares have been considered only recently. Da~Lozzo et al.~\cite{DBLP:conf/isaac/LozzoDEJ17} proved that every $K_{3,1,1,}$-free partial 2-tree admits a proper contact representation with homothetic squares, where a contact between two squares is \emph{proper} if they intersect in a line segment of positive length (in particular, proper contacts yield a strong contact representation). Eppstein~\cite{11011110} indicated that another family of graphs, defined recursively, can also be represented as a proper contact graph of squares. We remark that Klawitter et al.~\cite{DBLP:conf/gd/KlawitterNU15} proved that every triangle-free planar graph is the proper contact graph of (non-homothetic) axis-aligned rectangles.

\begin{figure}[h]	
	\centering
	\begin{subfigure}[b]{0.2\textwidth}
	    \centering
		\includegraphics[width=\textwidth, keepaspectratio]{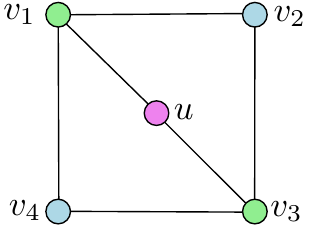}
		\caption{}\label{fig:graphs-a}
	\end{subfigure}
	\hfill
	\begin{subfigure}[b]{0.2\textwidth}
	    \centering
		\includegraphics[width=\textwidth, keepaspectratio]{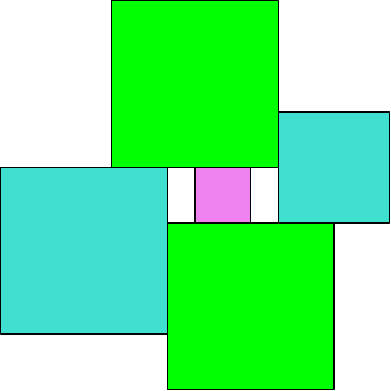}
	\end{subfigure}
	\hfill
	\begin{subfigure}[b]{0.2\textwidth}
		\centering
		\includegraphics[width=\textwidth, keepaspectratio]{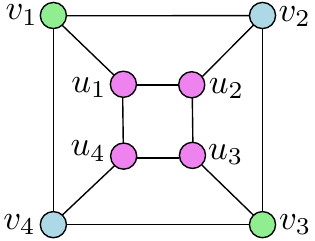}
		\caption{}\label{fig:graphs-b}
	\end{subfigure}
	\hfill
	\begin{subfigure}[b]{0.2\textwidth}
		\centering
		\includegraphics[width=\textwidth, keepaspectratio]{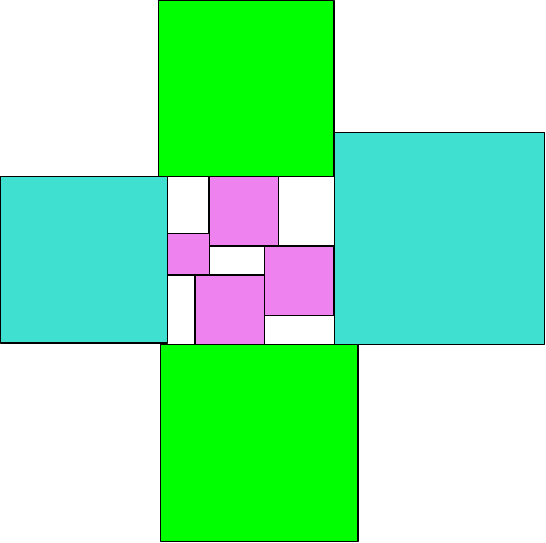}
	\end{subfigure}
	\caption{The two operations used to obtain a graph in $\mathcal{G}$ and their square contact representations.}\label{fig:graphs}
\end{figure}


\subparagraph{Contribution.}
Let $\mathcal{G}$ be a family of plane bipartite graphs defined recursively as follows. 
(i) $\mathcal{G}$ contains the 4-cycle $C_4$. 
(ii) If $G\in\mathcal{G}$ and $f=(v_1,v_2,v_3,v_4)$ is a bounded 4-face of $G$, then $\mathcal{G}$ also contains the graphs $G_a$ and $G_b$ obtained by the following two operations: (a) insert a vertex $u$ into $f$ and connect it to $v_1$ and $v_3$; (b) insert four vertices $u_1,\ldots , u_4$ into $f$, add the cycle $(u_1,u_2,u_3,u_4)$ and the edges $u_iv_i$ for $i=1,\ldots, 4$; see Fig.~\ref{fig:graphs}.

Every maximal 2-degenerate bipartite plane graph can be constructed by operation (a); and the 1-skeleton of every polycube whose dual graph is a tree~\cite{11011110} can be constructed by operation (b). However, the two operations jointly produce a larger class $\mathcal{G}$, which belongs to the class of 3-degenerate bipartite plane graphs. In a square contact representation (\SCR) of a graph in $\mathcal{G}$, every vertex $v_i$ corresponds to an axis-aligned square $s(v_i)$, and every bounded face to an axis-aligned rectangle $g(f_i)$, which is also called the \emph{gap} corresponding to $f_i$. We present our main result:

\begin{theorem}\label{thm:main}
Every graph in $\mathcal{G}$ admits a proper square contact representation. 
\end{theorem}

We prove Theorem~\ref{thm:main} by induction in Section~\ref{sec:scr}. For the induction hypothesis we establish a stronger version of the theorem in which one specifies intervals for the aspect ratios (defined as height/width) of every gap in the representation, then recursively creates the \SCR around those gaps. 

\begin{restatable}{theorem}{thmtwo}\label{thm:2}
Let $G\in \mathcal{G}$ be a graph with $n$ vertices and $n-3$ bounded faces $f_1,\ldots, f_{n-3}$. For all $\alpha_1,\ldots, \alpha_{n-3}>0$ and for all $\varepsilon>0$, the graph $G$ admits a proper square contact representation such that the aspect ratio of the gap corresponding to $f_i$ is $\alpha_i'$, with $|\alpha_i-\alpha_i'|<\varepsilon$, for all $i=1,\ldots , n-3$.
\end{restatable}

\begin{figure}[htbp]
 \centering
 \includegraphics{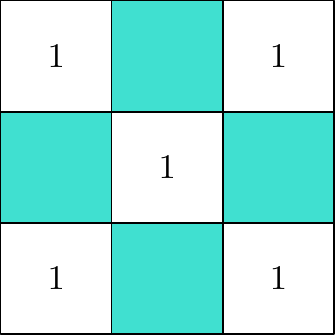}
 \caption{If all the gaps have aspect ratio 1, then scaling any of the squares to changing the point contacts into proper contacts would change the aspect ratios of the outer gaps.}
    \label{fig:square ratios}
\end{figure}

Figure~\ref{fig:square ratios} shows an example where the aspect ratios of the gaps cannot be specified exactly in a proper contact representation.

However, it turns out that $\mathcal{G}$ includes graphs that must be bounded by a rectangle whose aspect ratio is arbitrarily close to any given value, if they are inserted into a face of another graph in $\mathcal{G}$.

\begin{restatable}{theorem}{arbitraryaspectratios}\label{thm:arb}
For every $r, \delta>0$, there exists a bipartite plane graph $G\in \mathcal{G}$ with a 4-cycle as its outer face such that in every \SCR of $G$, the aspect ratio of the central gap between the four squares corresponding to that 4-cycle is confined to the interval $(r-\delta, r+\delta)$.
\end{restatable}

\subparagraph{Relation to rectangle tilings.} 
Theorem~\ref{thm:2} implies a tiling of a bounding box, where the tiles are squares (of aspect ratio 1) and rectangular gaps whose aspect ratios are prescribed up to an $\varepsilon$ error term. Note that the contact graph of this tiling, including squares and gaps, and four additional vertices for the four sides of the outer frame, is a triangulation.  Schramm~\cite{Schramm-Squares} (see also~\cite[Chap.~6]{Lovasz}) showed that for every inner triangulation $G$ of a 4-cycle without separating triangles there exists a rectangle contact representation of $G$ in which the rectangles have prescribed aspect ratios. However, some of the contacts between rectangles might be point contacts, and the interior of some of the separating 4-cycles may degenerate to a point. In the recursive construction of $\mathcal{G}$, step (ii) creates five separating 4-cycles in the triangulation of the tiling, one for each gap  (see Fig.~\ref{fig:separating-cycles}). In particular, if all five gaps degenerate to a point, then Schramm's result becomes trivial, but would not imply Theorem~\ref{thm:2}. The class of graphs defined in this paper is perhaps the first interesting case for which Schramm's approach is infeasible, as it cannot guarantee that the rectangles on the interior of the separating 4-cycles do not degenerate.

\begin{figure}[hp]	
	\centering
	\begin{subfigure}[b]{0.2\textwidth}
	    \centering
		\imagebox{0.2\textwidth}{\includegraphics[width=\textwidth, keepaspectratio]{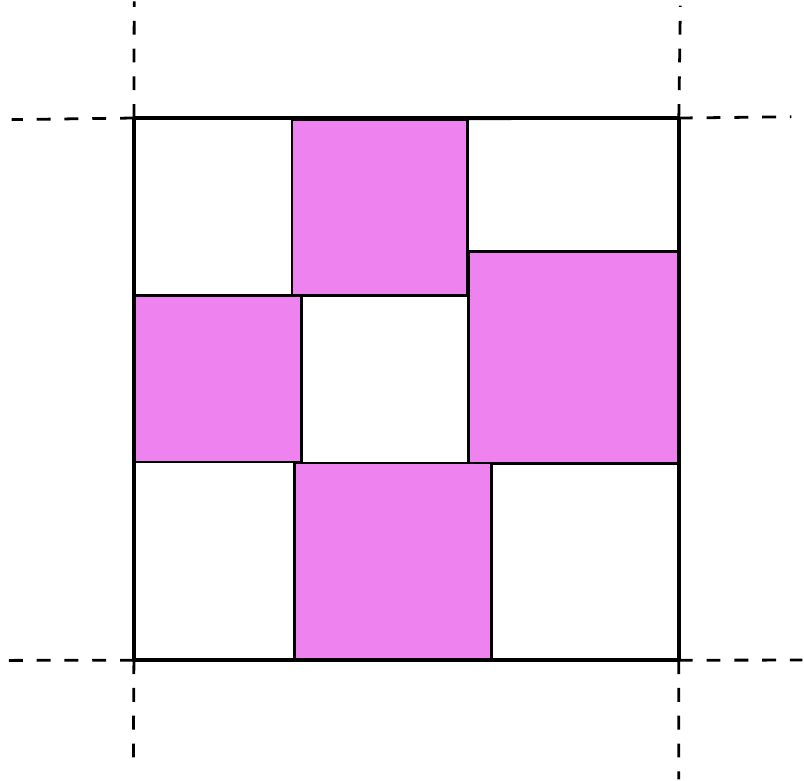}}
	\end{subfigure}
    \hspace{2cm}
	\begin{subfigure}[b]{0.2\textwidth}
	    \centering
		\imagebox{0.2\textwidth}{\includegraphics[width=\textwidth, keepaspectratio]{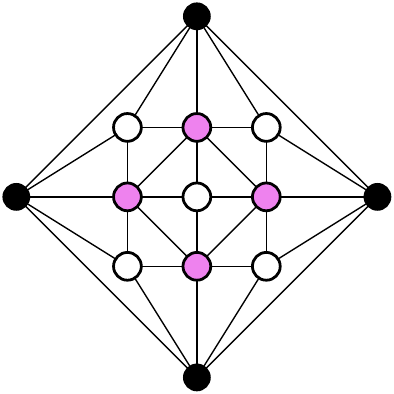}}
	\end{subfigure}
	\vspace{2.1cm}
	\caption{Left: a rectangular tiling with 9 tiles.  Right: the corresponding triangulation, where the outer 4-cycle corresponds to the four edges of the outer frame. }\label{fig:separating-cycles}
\end{figure}

\subparagraph{Outlook.} An obvious open problem is whether every triangle-free plane graph admits a proper square contact representation. Motivated by Schramm's results, one can also ask whether Theorem~\ref{thm:main} generalizes to the setting where each vertex of the graph is associated with an axis-aligned rectangle of given aspect ratio. 

\subparagraph{Terminology.} 
Let $G=(V,E)$ be an edge-maximal plane bipartite graph. In a square contact representation, every vertex $v_i$ corresponds to an axis-aligned square $s(v_i)$, and every bounded face to an axis-aligned rectangle $g(f_i)$, which is also called the \emph{gap} corresponding to $f_i$. The aspect ratio of an axis-aligned rectangle $r$ is $\mathrm{height}(r)/\mathrm{width}(r)$. The side length of a square $s$ is denoted by $\len(s)$. Scaling up a square from a corner by (or to) $x$ means to increase the width and height of the square by $x$ (or to $x$) in such a way that the position of the specified corner remains fixed.

\section{Maintaining a Square Contact Representation}

In this section, we show how to maintain a square contact representation of a graph in $\mathcal{G}$ under operations (a) and (b). Specifically, we show that one can insert one or four new squares corresponding to these operations in a rectangular gap of suitable size.  The following Lemmas are used in the proof of Theorem~\ref{thm:2} to recursively construct a \SCR for any given graph in $\mathcal{G}$.

\begin{lemma}\label{lem:1}
For every $\alpha,\beta>0$, there exists an axis-aligned rectangle that can be subdivided by two horizontal (resp., vertical) lines into three rectangles of aspect ratios $\alpha$, 1, and $\beta$, respectively.
\end{lemma}
\begin{proof}
Let $R$ be a rectangle of aspect ratio $\alpha + \beta + 1$, with width $x$ and height $(\alpha+\beta+1)x$. Two horizontal lines at distance $\alpha x$ and $\beta x$ from the top and bottom side of $R$, resp., subdivide $R$ into rectangles of aspect ratios $\alpha$, $1$, and $\beta$, as required; see Fig.~\ref{fig:Rectangle Subdivision}.
\end{proof}

\begin{figure}[htbp]
 \centering
 \includegraphics{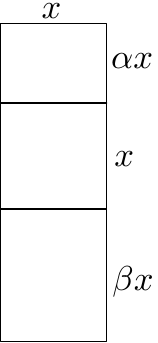}
 \caption{Constructing an outer rectangle given two inner rectangle aspect ratios.}
    \label{fig:Rectangle Subdivision}
\end{figure}

To establish Theorem~\ref{thm:main}, we need a stronger version of Lemma~\ref{lem:1} that allows the aspect ratios to vary within a small threshold.

\begin{lemma}\label{lem:1+}
For every $\alpha,\beta,\varepsilon>0$, there exists a $\delta>0$ such that any rectangle of aspect ratio $\gamma$ with $|\gamma-(\alpha+\beta+1)|<\delta$ can be subdivided by two horizontal lines into rectangles of aspect ratios $\alpha'$, 1, and $\beta'$ such that $|\alpha'-\alpha|<\varepsilon$ and $|\beta'-\beta|<\varepsilon$.
%
%
\end{lemma}
\begin{proof}
Let $\delta=\min\{\alpha,\beta,1,\varepsilon\}$. Let $R$ be a rectangle of aspect ratio $\gamma$, where $|\gamma-(\alpha+\beta+1)|<\delta$, with width $x$ and height $\gamma x$. 
Two horizontal lines at distance $\alpha x$ and $(1+\alpha)x$ from the top side of $R$ subdivide $R$ into rectangles of aspect ratios $\alpha$, $1$, and $\beta'=\gamma-\alpha-1$. Note that $\beta'>0$ and $|\beta'-\beta|=|\gamma-(\alpha+\beta+1)|<\delta\leq \varepsilon$.
\end{proof}

\begin{restatable}{lemma}{lemtwo}\label{lem:2}
For every $\alpha_1,\ldots, \alpha_5>0$, there exists an axis-aligned rectangle $R$ that can be subdivided into four squares and five rectangular gaps of aspect ratios $\alpha_1,\ldots ,\alpha_5$ such that (refer to Figs.~\ref{fig:graphs-b} and \ref{fig:configurations})
\begin{itemize}
    \item the four squares are each in contact with a side of $R$, and their contact graph is a 4-cycle (but the contacts along the 4-cycle are not necessarily proper);
    \item the first four gaps are each incident to the top-left, bottom-left, bottom-right, and top-right corner of $R$, respectively, and the fifth gap lies in the interior of $R$.
\end{itemize}
\end{restatable}

The proof of Lemma~\ref{lem:2} requires some preparation, and is presented later in this section. For convenience, we will rename $\alpha_1,\ldots,\alpha_5$ respectively based on the positions of the gaps to which they correspond as $\alpha_c$ (center), \aul (top-left), \aur (top-right), \adr (bottom-right), \adl (bottom-left). Also, name the squares incident to the top, bottom, right, and left side of $R$ as \su, \sd, \sr, and \slf, respectively.

We will prove Lemma~\ref{lem:2} by starting with an initial configuration (Fig.~\ref{fig:initial-config}), where the aspect ratio of the center gap is already $\alpha_c$, and there are improper contacts between adjacent squares of the cycle. Then we incrementally modify the configuration, while the center gap remains fixed, until all remaining gaps have the target aspect ratios \aul, \aur, \adr, and \adl. 
We denote the current aspect ratios of these gaps by \gul, \gur, \gdr, and \gdl in the same fashion as \aul$ ,\ldots, $\adl.   
We next define the initial configuration and four additional special configurations that play a role in intermediate steps of the incremental construction. 

\subparagraph{Initial configuration.}
To create the initial configuration, we start by drawing the interior gap and placing \su, $\ldots$, \slf incident to it, with each of their side lengths equal to the side of the interior gap to which they are incident (see Fig.~\ref{fig:initial-config}).
Note that the aspect ratios of every outer gap is $\alpha_c^{-1}$ in this configuration.

\begin{figure}[htbp]
 \centering
 \includegraphics{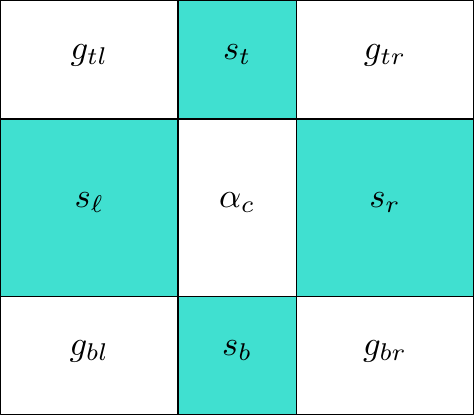}
 \caption{The initial configuration, with squares and gap aspect ratios labeled.}
    \label{fig:initial-config}
\end{figure}

\subparagraph{Pinwheel configuration.} A \emph{clockwise} pinwheel configuration is defined as follows (see Fig.~\ref{fig:1a}):
\begin{itemize}
    \item the bottom-right corner of \su lies on the left side of \sr,
    \item the bottom-left corner of \sr lies on the top side of \sd,
    \item the top-left corner of \sd lies on the right side of \slf,
    \item the top-right corner of \slf lies on the bottom side of \su.
\end{itemize}
A \emph{counterclockwise} pinwheel can be obtained by a reflection.

\begin{figure}[h]	
	\centering
	\vspace{\baselineskip}
	\begin{subfigure}[t]{2in}
		\centering
		\includegraphics{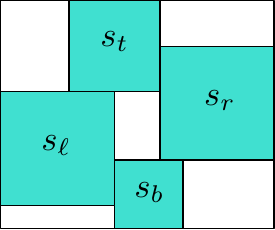}
		\caption{Clockwise Pinwheel}\label{fig:1a}		
		 \vspace{\baselineskip}
	\end{subfigure}
	\begin{subfigure}[t]{2in}
		\centering
		\includegraphics{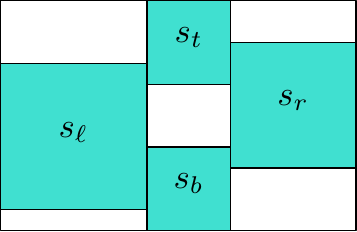}
		\caption{Vertical Stacked}\label{fig:1b}
       \vspace{\baselineskip}
	\end{subfigure}
	\vspace{\baselineskip}
	\begin{subfigure}[t]{2in}
		\centering
		\includegraphics{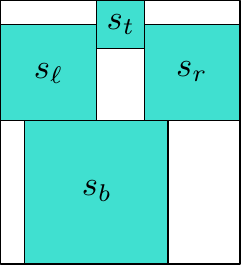}
		\caption{Downward Arrow}\label{fig:1c}
	\end{subfigure}
	\begin{subfigure}[t]{2in}
		\centering
		\includegraphics{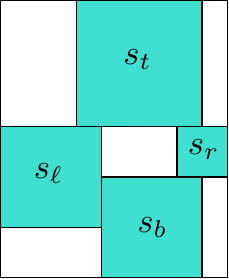}
		\caption{Clockwise Near-Pinwheel with reversed contact between \sr and \su}\label{fig:1d}
	\end{subfigure}
	\caption{Examples of four special configurations.}\label{fig:configurations}
\end{figure}

\subparagraph{Stacked configuration.} We define a \emph{vertical} stacked configuration as follows (see Fig.~\ref{fig:1b}):
\begin{itemize}
    \item the top-right corner of \sd lies on the left side of \sr,
    \item the top-left corner of \sd lies on the right side of \slf,
    \item the bottom-right corner of \su lies on the left side of \sr,
    \item the bottom-left corner of \su lies on the right side of \slf.
\end{itemize}                 
A \emph{horizontal} stacked configuration can be obtained by a $90^\circ$ rotation.

\subparagraph{Arrow configuration.} We define a \emph{downward} arrow configuration as follows (see Fig.~\ref{fig:1c}):
\begin{itemize}
    \item the top-right corner of \sd lies on the bottom side of \sr,
    \item the top-left corner of \sd lies on the bottom side of \slf,
    \item the bottom-right corner of \su lies on the left side of \sr,
    \item the bottom-left corner of \su lies on the right side of \slf.
\end{itemize}
\emph{Upward, leftward}, and \emph{rightward} arrow configurations can be obtained by rotation.  We also define the \emph{directional} square of the arrow configuration to be the one furthest in the direction after which the configuration is named (e.g., \sd for a downward arrow configuration).

\subparagraph{Near-pinwheel configuration.} We define a \emph{clockwise} near-pinwheel configuration as a configuration which would be a clockwise pinwheel configuration if one of the contacts between squares was changed from vertical to horizontal, or vice-versa (see Fig.~\ref{fig:1d}). This contact is called the \emph{reversed} contact of the near-pinwheel configuration.  A \emph{counterclockwise} near-pinwheel configuration can be obtained by reflection.

Lemmas~\ref{lem:slide-scale}--\ref{lem:near-pinwheel-resolution} below concern transformations of these special configurations, and are used in the proof of Lemma~\ref{lem:2}.

\begin{lemma}\label{lem:slide-scale}
Assume that the top-left corner of \sr is on the right side of \su and the bottom-left corner of \sr is on the right side of \sd, and let \aur$>$\gur be given. There exists a $d>0$ such that if we slide \sr upward by $d$ and scale it up by a factor of $d/g_{br}$ from its bottom-left corner, then no aspect ratio other than \gur changes, and 
after the transformation we have \aur = \gur, or \aur$>$\gur and
\sr and \sd have a point contact.  Similar statements hold after reflections and rotations of the configuration.
\end{lemma}
\begin{proof}
Let the bottom-right gap have height $h_1$ and width $w$ prior to the transformation. Assume that we slide \sr upward by some amount $d>0$ and scale it up by a factor of $d/g_{br}$ from its bottom-left corner.  After the transformation, it has height $h_1 + d$ and width $w + \frac{dw}{h_1}$. 
As
\[\frac{h_1}{w} = \frac{h_1+d}{w+\frac{dw}{h_1}},\]
the aspect ratio of the bottom-right gap has not changed. Let the height of top-right gap be $h_2$ prior to the transformation, and note that its width is also $w$.  After the transformation, it has height $h - d$ and width $w + \frac{d}{g_{br}}$.  Thus, its height monotonically decreases in $d$, and its width monotonically increases in $d$, so \gur monotonically decreases in $d$.  We can choose $d=\min(d_1, d_2)$, where $d_1\geq 0$ is the value which would reduce the contact between \sr and \sd to a single point after the transformation, and $d_2\geq 0$ is the value which would achieve \aur = \gur.
\end{proof}

\begin{lemma}\label{lem:pinwheel-resolution}
A clockwise (counterclockwise) pinwheel configuration can be transformed such that \gdr or \gul (\gur or \gdl) increases to, or such that \gur or \gdl (\gdr or \gul) decreases to any amount $\gamma>0$, while all other aspect ratios remain the same.
\end{lemma}
\begin{proof}
Assume w.l.o.g.\ that we are given a \emph{clockwise} pinwheel configuration, and we wish to increase the aspect ratio \gdr to $\gamma>$\gdr.
If we scale up \sd from its top-left corner by some amount $d_1$, then \gdl will increase.  To account for this change, though, we can scale up \slf as well so that \gdl remains constant. Let $h$ be the height of the central gap.  Then, 
\[g_{bl} = \frac{\len(s_b) - \len(s_\ell) + h}{\len(s_\ell)}.\]  
After increasing the length of \sd by $d_1$, we must then increase the length of \slf by some amount $d_2$ such that
\[\frac{\len(s_b) - \len(s_\ell) + h}{\len(s_\ell)} = \frac{(\len(s_b) + d_1) - (\len(s_\ell) + d_2) + h}{\len(s_\ell) + d_2}\]
so that \gdl does not change.  Solving this equation for $d_2$ yields
\[d_2 = d_1 \frac{\len(s_\ell)}{\len(s_b) + h}.\]
Because \slf is not in contact with the bottom of $R$, $\len(s_\ell) < \len(s_b) + h$.  Thus, $d_2 < d_1$.

Let $w$ be the width of the central gap.  Then,
\[g_{tl} = \frac{\len(s_t)}{\len(s_\ell) - \len(s_t) + w}.\] 
After increasing the length of \slf by $d_2$, to maintain \gul, we must increase the length of \su by some amount $d_3$ such that

\[\frac{\len(s_t)}{\len(s_\ell) - \len(s_t) + w} = \frac{\len(s_t) + d_3}{(\len(s_\ell) + d_2) - (\len(s_t) + d_3) + w}.\] 
Solving for $d_3$ gives

\[d_3 = d_2 \frac{\len(s_t)}{\len(s_\ell) + w}.\]
Because \su is not in contact with the left side of $R$, $\len(s_t) < \len(s_\ell) + w$.  Thus, $d_3 < d_2$.

After increasing the length of \su by $d_3$, we must increase the length of \sr by some amount $d_4$ to maintain \gur. Similarly to the argument above, we obtain $d_4 < d_3$, and thus, $d_4 < d_1$.

So, this series of transformations, preserving \gur, \gul, \gdl, and the central gap, increases the length of \sd by $d_1$, which is more than the amount it increases the length of \sr, $d_4$.  Specifically, 
\[d_4 = \frac{d_1\len(s_\ell)\len(s_t)\len(s_r)}{(\len(s_b) + h)(\len(s_\ell) + w)(\len(s_t) + h)} < d_1.\]
Before the transformations, the top boundary of \sd overlapped the bottom boundary of \sr by some amount $x$.  After the transformations, it overlaps by $x + d_1$, because \sd has been scaled up from its top-left corner.  

The width of the bottom-right gap equals $\len(s_r)$ minus the length of the common boundary between \sr and \sd.  Because the length of that common boundary increases by $d_1$, but $\len(s_r)$ increases only by $d_4 < d_1$, the width decreases. Consequently, the width of the bottom-right gap decreases and its height increases linearly in $d_1$. Overall, \gdr monotonically increases in $d_1$. We have constructed a series of transformations that can increase \gdr to any $\gamma > g_{br}$ with a suitable $d_1$. 
\end{proof}

\begin{lemma}\label{lem:stacked-resolution}
A vertical (resp., horizontal) stacked configuration with a point contact between two of the squares can be transformed such that the aspect ratio of the outer gap between those squares increases (resp., decreases) to any amount $\gamma>0$ while all other aspect ratios remain the same.
\end{lemma}
\begin{proof}
Assume w.l.o.g.\ that we are given a vertical stacked configuration in which \sr and \sd have a point contact, and we wish to increase the aspect ratio \gdr to $\gamma>$\gdr.

If there is not a point contact between \slf and \su, then the following transformation can be applied. Scale up \sd from its top-left corner to increase \gdr. To account for the resulting change in \gdl, scale up \slf and translate it downward while maintaining \gul, as described in Lemma~\ref{lem:slide-scale}. This transformation will either increase \gdr to $\gamma$, or it will result in a point contact between \slf and \su.

If there is a point contact between \slf and \su, then the squares are arranged in a pinwheel configuration, and by Lemma~\ref{lem:pinwheel-resolution} we can increase \gdr to $\gamma$ while maintaining all other aspect ratios.
\end{proof}

\begin{lemma}\label{lem:arrow-resolution}
An upward or downward (resp., rightward or leftward) arrow configuration, with a point contact between the directional square and one of its neighbors, can be transformed such that the aspect ratio of the outer gap between those squares increases (resp., decreases) to any amount $\gamma>0$ while all other aspect ratios remain the same.
\end{lemma}
\begin{proof}
Assume w.l.o.g.\ that we are given a downward arrow configuration in which \sr and \sd have a point contact, and we wish to increase the aspect ratio \gdr to $\gamma$.

If \sd and \slf do not have a point contact, translate \sd to the right while scaling it up in order to maintain \gdl (as described in~Lemma \ref{lem:slide-scale}) while increasing \gdr until \gdr = $\gamma$, or until there is a point contact between \sd and \slf.

If \sd and \slf have a point contact, then scale up \sd from its top-left corner to increase \gdr.  To account for the corresponding change in \gdl, translate \slf downward while scaling it up to maintain \gul (as described in Lemma~\ref{lem:slide-scale}) until \gdr = $\gamma$, or until there is a point contact between \slf and \su.

If \slf and \su have a point contact, then the squares are arranged in a pinwheel configuration, and by Lemma~\ref{lem:pinwheel-resolution} we can increase \gdr to $\gamma$ while maintaining all other aspect ratios.
\end{proof}

\begin{lemma}\label{lem:near-pinwheel-resolution}
A near-pinwheel configuration can be transformed such that the aspect ratio of the outer gap in the direction of the near-pinwheel (clockwise or counterclockwise) from the reversed contact increases to any amount $\gamma>0$ if its left side is the side of a square, or decreases to any amount $\gamma>0$ if its top side is the side of a square, while all other aspect ratios remain the same.
\end{lemma}
\begin{proof}
Assume w.l.o.g.\ that we are given a clockwise near-pinwheel with a reversed top-right contact (as in Figure~\ref{fig:1d}), and we wish to increase the aspect ratio \gdr to $\gamma$.

Perform the following transformation until \su and \sr have a point-contact or until \gdr has been increased to $\gamma$.  Scale up \sd from its top-left corner by some amount. To account for the corresponding change in \gdl, scale up \slf from its top-right corner.  To account for the corresponding change in \gul, scale up \su and translate it to the left while maintaining \gur as described in Lemma~\ref{lem:slide-scale}.

If \gdr does not reach its target value once \su and \sr have a point contact, then the configuration is a pinwheel, and by Lemma~\ref{lem:pinwheel-resolution} we can increase \gdr to $\gamma$.
\end{proof}

We now have everything needed to prove Lemma~\ref{lem:2}.

\begin{proof}[Proof of Lemma~\ref{lem:2}] 
Let $\alpha_c$, \aul, \aur, \adr, and \adl be given.  Start with the initial configuration (cf.~Fig.~\ref{fig:initial-config}). If the target aspect ratios of all four outer gaps are $\alpha_c^{-1}$, then $R$ can be drawn now with aspect ratio $\alpha_c$. Otherwise, one or more of the outer gaps must have their aspect ratios changed, either by increasing or decreasing them.  

Rotate and reflect the initial configuration if necessary such that at least one gap needs to be made wider (i.e., $\alpha<g$), and the ratio $g/\alpha$ is maximal for the top-right gap. In order to change \gur to \aur, we can scale up \sr from its bottom-left corner until \gur $=$ \aur.  This scaling will not affect \gul or \gdl, but it will decrease \gdr. After the scaling, the bottom-right gap will either have the target aspect ratio already, need to be wider yet, or need to be narrower. From now on, we will not mention the case where a gap has reached its target aspect ratio already, because it just means that the next step can be skipped.

If the bottom-right gap needs to be wider yet, then by Lemma~\ref{lem:slide-scale} we can scale up \sr and translate it downward until \gdr = \adr without changing \gur. As $g/\alpha$ is assumed to be maximal for the top-right gap, if this transformation results in a point contact between \sr and \su, it also achieves \gdr = \adr (because otherwise, \gdr $>$ \gur = \aur).

If the bottom-right gap needs to be narrower, then we can scale up \sd from its top-left corner until \gdr = \adr.  This will increase \gdl.

Now, we can assume that \gur = \aur and \gdr = \adr.  We distinguish between four cases:
\begin{enumerate}
    \item \sd has not been scaled, and either \adl $\leq \alpha_c^{-1}$ or \aul $\leq \alpha_c^{-1}$.
    \item \sd has been scaled up from its top-left corner, \adl $\leq$ \aul, and \adl $\leq \alpha_c^{-1}$.
    \item \sd has been scaled up from its top-left corner, \aul $\leq$ \adl, and \aul $\leq \alpha_c^{-1}$.
    \item \aul $> \alpha_c^{-1}$ and \adl $> \alpha_c^{-1}$.
\end{enumerate}

\textbf{Case~1}: \sd has not been scaled, and either \adl $\leq \alpha_c^{-1}$ or \aul $\leq \alpha_c^{-1}$.  Reflect the configuration, if necessary, such that \adl $\leq$ \aul.  Scale up \slf from its top-right corner until \gdl = \adl (making the top-left gap wider).  Then, if \gul needs to decrease further, by Lemma~\ref{lem:slide-scale} we can scale up and translate \slf until \gul = \aul to achieve all target aspect ratios (once again, this transformation guarantees \gul = \aul even if it results in a point contact, because we assume \adl $\leq$ \aul). Otherwise, the top-left gap needs to be narrower. Since the configuration is a horizontal stacked configuration, and by Lemma~\ref{lem:stacked-resolution} we can apply a series of transformations to achieve all target aspect ratios.

\textbf{Case~2}: \sd has been scaled up from its top-left corner, \adl $\leq$ \aul, and \adl $\leq \alpha_c^{-1}$.  Scale up \slf from its top-right corner until \gdl = \adl.  This transformation decreases \gul.  Then, if \gul needs to decrease further, by Lemma~\ref{lem:slide-scale} we can scale up and translate \slf until \gul = \aul to achieve all target aspect ratios (once again guaranteed because \adl $\leq$ \aul).  Otherwise the top-left gap needs to be narrower. Since the squares are arranged in a pinwheel configuration, Lemma~\ref{lem:pinwheel-resolution} completes the proof.

\textbf{Case~3}: \sd has been scaled up from its top-left corner, \aul $\leq$ \adl, and \aul $\leq \alpha_c^{-1}$.  Scale up \slf from its bottom-right corner until \gul = \aul.  This transformation decreases \gdl.  Then, if \gdl needs to decrease further, by Lemma~\ref{lem:slide-scale} we can scale up \slf and translate it downward, maintaining all other aspect ratios, until \gdl = \adl or \slf and \su have a point contact.  If \slf and \su have a point contact, then the squares are arranged in a pinwheel configuration, and Lemma~\ref{lem:pinwheel-resolution} completes the proof.
Otherwise, \gdl needs to increase. Since the squares form a downward arrow configuration in this case, with a point contact between \sd and \slf, Lemma~\ref{lem:arrow-resolution} completes the proof.

\textbf{Case~4}: \aul $> \alpha_c^{-1}$ and \adl $> \alpha_c^{-1}$.
We distinguish between two subcases. 

\textbf{Case~4.1}: 
If the top-right corner of \sd lies on the bottom side of \sr, then by Lemma~\ref{lem:slide-scale}, we can translate \sd to the left while scaling it up until \gdl = \adl or \sd and \sr have a point-contact, while maintaining all other aspect ratios. 
If \gdl = \adl, then the configuration is a near-pinwheel and Lemma~\ref{lem:near-pinwheel-resolution} completes the proof.  Otherwise, if \sd and \sr have a point-contact, then the conditions of Case~4.2 below are satisfied and we proceed as follows.

\textbf{Case~4.2}: 
If the top-right corner of \sd lies on the left side of \sr, then scale up \su from its bottom-right corner until \gul = \aul and scale up \sd from its top-right corner until \gdl = \adl.  Now, \gur and \gdr (which were previously at their target values) both need to decrease. Reflect the configuration, if necessary, so that the width of the bottom-right gap needs to be increased by a larger amount than the top-right gap.  Scale up \sr from its bottom-left corner until \gur = \aur.  Then, because the width of the bottom-right gap needed to be increased by a larger amount of the two, it still needs to be wider.  The configuration is a rightward arrow, so by Lemma~\ref{lem:arrow-resolution}, we can decrease \gdr arbitrarily while maintaining the other aspect ratios.
\end{proof}

The following lemma, Lemma~\ref{lem:2+}, shows that all improper contacts can be replaced by proper contacts at the expense of allowing the five aspect ratios to vary within a given threshold. Using exact values of the aspect ratios, Lemma~\ref{lem:2} can only guarantee single-point contacts.  However, it is easy to extend Lemma~\ref{lem:2} to Lemma~\ref{lem:2+} by changing any improper contacts among adjacent squares in the 4-cycle into proper contacts. 

\begin{lemma}\label{lem:2+}
For every $\alpha_1,\ldots, \alpha_5>0$ and $\varepsilon>0$, there exists a $\lambda>0$ and a $\delta>0$ such that every axis-aligned rectangle $R$ of aspect ratio $\lambda'$, $|\lambda-\lambda'|<\delta$, can be subdivided into four squares and five gaps of 
aspect ratios $\alpha_i'$, with $|\alpha_i'-\alpha_i|<\varepsilon$, for $i=1,\ldots, 5$ such that 
\begin{itemize}
    \item the four squares are each in contact with a side of $R$, and their contact graph is a 4-cycle, and all contacts are proper;
    \item the first four gaps are each incident to the top-left, bottom-left, bottom-right, and top-right corner of $R$, respectively, and the fifth gap lies in the interior of $R$.
\end{itemize}
\end{lemma}
\begin{proof}
Let $\alpha_c$, \aul, \aur, \adr, \adl, and $\varepsilon>0$ be given. 
By Lemma~\ref{lem:2}, there is a rectangle $R$ with some aspect ratio 
$\lambda$ that can be subdivided into five gaps and four squares $s_b$, $s_t$, $s_\ell$, and $s_r$ whose contact graph is a cycle. 

\textbf{Case~1.}
Assume first that all four contacts in the cycle are proper. Then Lemma~\ref{lem:2+} holds with the same $\lambda$. In each case, there exists a square that can be scaled up or down while maintaining proper contacts in the cycle. When scaling a single square, the aspect ratio of the bounding box $R$ and some of the gaps change continuously. By continuity, there exists a $\delta>0$ such that if the aspect ratios of the bounding box is $\lambda'$ with $|\lambda'-\lambda|<\delta$, then all five gaps are at most $\varepsilon$ from their target values.

\textbf{Case~2.}
Next assume that one or more contacts in the cycle are improper, i.e., two squares intersect in a common corner. For each improper contact, we can successively scale up one of the two squares to establish a proper contact. 
We scale up each square by a sufficiently small amount such that the aspect ratios of the five gaps change by less than $\varepsilon/2$. Let $\lambda'$ be the aspect ratio of the new bounding box. We can show, similarly to Case~1, that Lemma~\ref{lem:2+} holds with $\lambda=\lambda'$ and some $\delta>0$ by continuity.
\end{proof}

\section{Proof of Theorem~\ref{thm:2}}\label{sec:scr}

Finally, we have all the tools needed to prove Theorem~\ref{thm:2}. We restate it for convenience:

\thmtwo*

\begin{proof}
We proceed by induction on $n$, the number of vertices of $G$.

\textbf{Basis step.} Assume that $G=C_4$ is a 4-cycle with a single bounded face $f_1$. It is clear that for any $\alpha_1>0$, $C_4$ has a proper square contact representation as a pinwheel configuration in which the gap corresponding to $f_1$ has aspect ratio $\alpha_1$.

\textbf{Induction step.} Let $G\in \mathcal{G}$ be a graph with $n\geq 5$ vertices, and assume that the claim holds for all graphs in $\mathcal{G}$ with fewer than $n$ vertices. Then $G$ was constructed from a graph $G_0\in \mathcal{G}$ with operation (a) or (b) that inserts one or four vertices into a 4-face $f_0=(v_1,\ldots , v_4)$. 
We may assume w.l.o.g. that $v_1$ and $v_3$ correspond to squares that lie on the vertical sides of the gap corresponding to $f_0$ in any square contact representation.
We distinguish between two cases.

\textbf{Case (a).} Assume that $G$ was obtained from $G_0$ by inserting a vertex $u$ into $f_0$ and connecting it to $v_1$ and $v_3$. This operation subdivides $f_0$ into $f_1$ and $f_2$; and all other faces are present in both $G$ and $G_0$. 
Let $\alpha_0=\alpha_1+\alpha_2+1$. 
By Lemma~\ref{lem:1+}, there exists a $\delta>0$ such that any rectangle of aspect ratio $\alpha_0'$ with $|\alpha_0'-\alpha_0|<\delta$ can be subdivided by two horizontal lines into rectangles of aspect ratios $\alpha_1'$, 1, and $\alpha_2'$ such that $|\alpha_1'-\alpha_1|<\varepsilon$ and $|\alpha_2'-\alpha_2|<\varepsilon$.
The induction hypothesis with $\varepsilon_0=\min\{\varepsilon,\delta\}$ implies that $G_0$ admits a proper square contact representation such that the gap corresponding to $f_0$ has aspect ratio $\alpha_0'$, where $|\alpha_0'-\alpha_0|<\varepsilon_0\leq\delta$, and all other gaps are at most $\varepsilon_0\leq\varepsilon$ off from their target aspect ratios. Lemma~\ref{lem:1+} now yields a subdivision of the gap corresponding to $f_0$ into a square in proper contact with the squares corresponding to $v_1$ and $v_3$, and two gaps of aspect ratios $\alpha_1'$ and $\alpha_2'$ with $|\alpha_1'-\alpha_1|<\varepsilon$ and $|\alpha_2'-\alpha_2|<\varepsilon$.

\textbf{Case (b).} Assume that $G$ was obtained from $G_0$ by inserting a 4-cycle $(u_1,u_2,u_3,u_4)$ into $f_0$ and adding the edges $u_i v_i$ for $i=1,\ldots , 4$. 
This operation subdivides $f_0$ into five faces $f_1,\ldots ,f_5$ of $G$; and all other faces are present in both $G$ and $G_0$.

By Lemma~\ref{lem:2+}, there exists an $\alpha_0>0$ and a $\delta>0$ such that any rectangle of aspect ratio $\alpha_0'$ with $|\alpha_0'-\alpha_0|<\delta$ can be subdivided into four squares and five gaps corresponding to $f_1,\ldots ,f_5$, of aspect ratios $\alpha_1',\ldots , \alpha_5'$, respectively, 
such that $|\alpha_i'-\alpha_i|<\varepsilon$ for $i=1,\ldots ,5$.
The induction hypothesis with $\varepsilon_0=\min\{\varepsilon,\delta\}$ implies
that $G_0$ admits a proper square contact representation such that the gap corresponding to $f_0$ has aspect ratio $\alpha_0'$, where $|\alpha_0'-\alpha_0|<\varepsilon_0\leq\delta$, and all other gaps are at most $\varepsilon_0\leq\varepsilon$ off from their target aspect ratios. Lemma~\ref{lem:2+} now yields a subdivision of the gap corresponding to $f_0$ into four squares, each in contact with a unique one of $v_1,\ldots, v_4$ and cyclically in contact with one another, and five gaps of aspect ratios $\alpha_1,\ldots,\alpha_5$ with $|\alpha_i'-\alpha_i|<\varepsilon$ for $i=1,\ldots ,5$.
\end{proof}

\section{Proof of Theorem~\ref{thm:arb}}
\label{app:arb}

\begin{restatable}{lemma}{spacing}\label{lem:spacing}
For every integer $n > 2$, $K_{2,n}\in \mathcal{G}$; and in any
\SCR of $K_{2,n}$, if the squares corresponding to the partite set of size two have
side lengths $\ell_1$ and $\ell_2$, then the distance between these squares is less than
$\frac{\min(\ell_1, \ell_2)}{n-2}$.

\end{restatable}
\begin{proof}
Let $s_1$ and $s_2$ be the squares of side lengths $\ell_1$ and $\ell_2$, respectively, in some \SCR of $K_{2,n}.$  W.l.o.g., we may assume that $\ell_1 \geq \ell_2$ and that $s_1$ is below $s_2$.  There exists a rectangle between $s_1$ and $s_2$ whose top side is the side of $s_2$ and whose height is the distance between $s_1$ and $s_2$.  It is clear that at most two of the $n$ squares corresponding to the other partite set can be anything but fully contained in this rectangle (see Figure~\ref{fig:corridor}).  Thus, the other $n-2$ squares must be inside this rectangle, and each must have the same side length because they contact the top and bottom of this rectangle.  Furthermore, the sum of their side lengths is less than $\ell_2$, because the squares don't overlap.  Thus, each of these squares has height less than $\frac{\ell_2}{n-2}$, and the distance between $s_1$ and $s_2$ is less than $\frac{\ell_2}{n-2}$.
\end{proof}

\begin{figure}[htbp]
	\centering
    \includegraphics{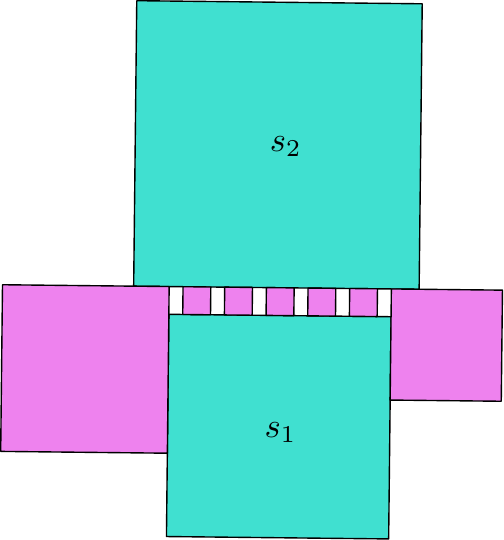}
	\caption{At most two squares are outside of the rectangle between $s_1$ and $s_2$.}\label{fig:corridor}
\end{figure}

\arbitraryaspectratios*
\begin{proof}
Let $r,\delta>0$ be given.
By applying a $90^\circ$ rotation, if necessary, we may assume that 
$r\in [1, \infty)$.

To construct $G$, we first construct its \SCR. We start with the 4-cycle, and successively insert squares into a remaining gap (defined below). After $i$ iterations, we obtain a graph $G_i$. We also maintain an interval $I_i$ such that $r\in I_i$ and in \emph{every} \SCR of $G_i$, the aspect ratio of the central  gap must be in $I_i$. Initially, we set $I_0=(0,\infty)$. 
We show that $I_{i+1}\subset I_i$ and $|I_i|<2^{-i}$ for all $i\in \mathbb{N}$. Consequently, $I_i\subset (r-\delta,r+\delta)$ when $2^{-i}<\delta$,
and we can return $G=G_i$. 

In each iteration, we repeatedly insert a square into a gap in the \SCR contacting either the top and bottom of the gap, or the left and right.  Clearly, the contact graph corresponding to the resulting \SCR will be a 2-degenerate plane bipartite graph. Whenever we insert a square into a gap, we will also assume that it contacts one additional side of the gap; however, instead of an actual contact, we can only guarantee that in any \SCR, they are sufficiently close to that side (cf.~Lemma~\ref{lem:spacing}): If the square contacts the left and right, then it must be very close to the bottom; if it contacts the top and bottom, then it must be very close to the left. Specifically, if $m$ is the total number of squares used in the rest of the construction, and $\ell$ is the side length of the largest square used in the construction, we can insert $\lceil\frac{2m\ell}{\delta} + 2\rceil$ squares in between each square and the side it is supposed to be close to.  This will ensure that each square is at most $\frac{\delta}{2m}$ apart from the side it is supposed to be close to, and thus that the aspect ratio differs from what the aspect ratio would be if these contacts actually existed by less than $\frac{\delta}{2}$. We can carry out the rest of the proof under the assumption that these squares in fact contact that side, and that the interval for the target aspect ratio is $(r-\frac{\delta}{2}, r+\frac{\delta}{2})$.

Because of these assumed additional contacts, there is always only one remaining gap in the course of the recursive construction. We will call this the \emph{remaining} gap.

Let the aspect ratio of the central gap (of the outer 4-cycle) be constrained to the interval $I_i$. When we insert a square into the remaining gap, either the lower or upper bound on this aspect ratio will become constrained to some $c\in I_i$.  Specifically, after inserting a square which contacts the left and right (and the bottom, as an additional contact) of the remaining gap, the lower bound increases to $c$ and the upper bound is unchanged.  This follows because the inserted square must be at least the width of the remaining gap, so the remaining gap's aspect ratio must be at least 1. However, it does not impose any constraint on the maximum height of the remaining gap, since the top of the square does not contact the top of the gap.  Similarly, after inserting a square which contacts the top and bottom (and left) of the remaining gap, the height of the gap is limited to the height of the square, so the remaining gap's aspect ratio must be at most 1, while the width of the gap is no further constrained.  As we will show later, the central gap's aspect ratio varies monotonically in the aspect ratio of the remaining gap.  Thus, we know that some $c$ must exist because inserting a square which contacts the top and bottom and inserting a square which contacts the left and right will each change a different one of the bounds of the aspect ratio of the remaining gap, and hence the central gap, to the same value.

So, one can always insert a sequence of squares contacting either the top and bottom or the left and right of the remaining gap, and it will either increase the lower bound or decrease the upper bound of the interval $I_i$, while containing the target aspect ratio $r$. In the remainder of the proof, we choose a specific sequence of insertions and show that both the upper and lower bounds  converge to $r$.

\subparagraph{Phases.}
Each iteration of the construction will consist of inserting squares into the remaining gap, $g$, in two \emph{phases}.  In each phase, we will either insert some number of squares which contact the left and right edges of the gap (a \emph{vertical} phase) or some number of squares with contact the top and bottom (a \emph{horizontal} phase).  The number of squares inserted is the \emph{size} of that phase.  Because the squares in each phase contact the same two sides of the gap, each phase will either increase the lower bound or decrease the upper bound of the interval $I_i$.  
W.l.o.g., let the next phase to insert be horizontal, setting some upper bound on the aspect ratio of $g$.  Then, by Lemma~\ref{lem:spacing}, we can insert a sufficiently large phase to reduce the distance between the last square in this horizontal phase and the side of $g$ to an arbitrarily small value, bringing the lower bound of the aspect ratio of $g$ arbitrarily close to the upper bound.  Because the central gap's aspect ratio varies monotonically in the aspect ratio of $g$, for any vertical (resp., horizontal) phase, there exists a $k$ for which inserting a phase of size $k+1$ would bring the lower (resp., upper) bound of $I_i$ above (resp., below) $r$.

We will use the following process to construct a \SCR whose central gap's aspect ratio is constrained to $(r-\delta, r+\delta)$, assuming $r \geq 1$.  Let the interval which is the bounds of the central gap's aspect ratio be $I_i=(a_i, b_i)$.  Starting with the four outer squares, while $|I_i| \geq \frac{\delta}{2}$:

\begin{enumerate}
    \item Insert a vertical phase whose size is the largest possible such that $a_i \leq r$.
    \item Insert a horizontal phase whose size is the largest possible such that $b_i \geq r$.
\end{enumerate}

Let $n$ be the total number of iterations.

\subparagraph{Convergence.} 
It is clear from the construction that $r\in I_{i+1}\subset I_i$ for all $i\in \mathbb{N}$. It remains to show that $|I_i|\leq 2^{-i}$. 
To prove the convergence, we will construct the same \SCR from the inside-out.  We start with an arrangement which is just the remaining gap, a rectangle, and add phases of squares alternatively contacting the left and bottom of this arrangement, as shown in Figure~\ref{fig:backwards}.  After adding phases in this way, the four outer squares can be added so that this construction ends with the same \SCR as we constructed with the above process.

\begin{figure}[htbp]
	\centering
    \includegraphics{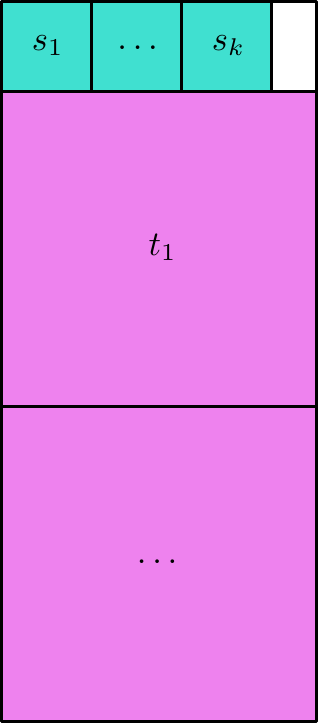}
	\caption{Starting with the remaining gap, we add a horizontal phase of squares (labeled $s_1, \ldots, s_k$ in this figure), then a vertical phase ($t_1, \ldots$), and will continue with alternating phases.}\label{fig:backwards}
\end{figure}

Let the width of the configuration after adding $i$ vertical phases be 1, and the height $h$ (and thus, the aspect ratio), be in some interval $J_i = (c, d)$. In particular, note that $J_0 = (0, \infty)$ and $J_n = I_n$.  We will then add a horizontal phase of size $k$, then a vertical phase of size $\ell$.  Note that each square in the horizontal phase has side length $h$, and each square in the vertical phase has side length $kh+1$.  Thus, the aspect ratio of the arrangement after inserting these phases is now

\[\frac{(kh+1)\ell + h}{kh+1}.\]  

This expression shows that $|J_i|<\infty$ for all $i\geq 1$.  By adding an extra iteration with $k_0, \ell_0 = 1$ we can also guarantee that $J_0 < 1$.
We can transform this expression as follows:

\[\frac{(kh+1)\ell + h}{kh+1} = \frac{(k\ell + 1)h + \ell}{kh+1} 
= \frac{k\ell + 1}{k} - \frac{\frac{1}{k}}{kh+1}.\]

This shows that the aspect ratio of the central gap varies monotonically in the aspect ratio of the remaining gap (as noted earlier).

As $\frac{k\ell + 1}{k}$ is a positive constant, it does not affect the length of the interval $J_{i+1}$.  Thus, we can say now that 

\begin{align*}
  |J_{i+1}| 
&< \frac{1}{k}\,  \left|\frac{1}{kd+1} - \frac{1}{kc+1}\right|\\
& < \frac{|c-d|}{(kc+1)(kd+1)}.
\end{align*}

We know that $d$ is at least 1, because $r \geq 1$, and $k$ is at least 1 as well.  Thus, the denominator is at least 2, and since $|J_i| = d - c$, 

\[|J_{i+1}| < \frac{|J_i|}{2}.\]

Combined with $|J_0|<1$, this implies $|I_n|=|J_n|<2^{-n}$,
and so $|I_n|<\delta$ if $2^{-n}<\delta$, or equivalently, $n>\log\delta^{-1}$.
\end{proof}

\subsection*{Acknowledgements}
The author would like to thank Stefan Felsner for bringing to his attention the connection to Schramm's square tiling result.  Additionally, he would like to thank Csaba Toth for all of his help and support in presenting these results.

\bibliography{Full_7-18-21_Bib}

\end{document}